\documentclass[aps,pra, nofootinbib, groupedaddress]{revtex4}

\usepackage{mathpazo}
\usepackage{amsmath}
\usepackage{amsfonts,amssymb,amsthm, bbm,braket}
\usepackage{graphicx}   
\usepackage{subfigure}  
\usepackage{amsbsy} 
\usepackage[bold]{hhtensor} 
\usepackage{dsfont}

\newcommand\herm{\mathbb H(\mathcal H)}

\newcommand\Tr{\mathrm{Tr}}

\newcommand\ip[2]{\langle#1 | #2\rangle}
\newcommand\op[2]{|#1\rangle\!\langle#2|}

\newcommand\abs[1]{\left|#1\right|}
\newcommand\av[1]{\left\langle #1 \right\rangle}
\newcommand\norm[1]{\|#1\|}
\newcommand\conj[1]{\overline{#1}}

\newtheorem{definition}{Definition}

\newtheorem{theorem}{Theorem}

\usepackage{hyperref} 

\hypersetup{
  colorlinks   = true, 
  urlcolor     = blue, 
  linkcolor    = blue, 
  citecolor   =  red 
}

\begin{document}

\title{Quantum Bochner's theorem for phase spaces built on projective representations}
\author{Ninnat Dangniam}
\affiliation{
Center for Quantum Information and Control,
University of New Mexico,
Albuquerque, New Mexico, 87131-0001}

\author{Christopher Ferrie}
\affiliation{
Center for Quantum Information and Control,
University of New Mexico,
Albuquerque, New Mexico, 87131-0001}
\affiliation{Centre for Engineered Quantum Systems, School of Physics, The University of Sydney, Sydney, NSW, Australia}

\begin{abstract}Bochner's theorem gives the necessary and sufficient conditions on a  function such that its Fourier transform corresponds to a true probability density function.  In the Wigner phase space picture, \emph{quantum} Bochner's theorem gives the necessary and sufficient conditions on a function such that it is a quantum characteristic function of a valid (and possibly mixed) quantum state and such that its Fourier transform is a true probability density.  We extend this theorem to discrete phase space representations which possess enough symmetry. More precisely, we show that discrete phase space representations that are built on projective unitary representations of abelian groups, with a slight restriction on admissible 2-cocycles, enable a quantum Bochner's theorem.
\end{abstract}

\date{\today}

\maketitle


\tableofcontents

\section{Introduction\label{S:Introduction}}

The differences between classical and quantum mechanics have distinct manifestations depending on how one looks at the problem.  A common approach which attempts to ground classical and quantum theory in the same picture is phase space. Phase space is a natural concept in classical theory since
it is equivalent to the state space. The idea of formulating quantum theory in
phase space dates back to 1932 with Wigner \cite{Wigner1932Quantum}. The now termed ``Wigner function'' is a quasi-probability
distribution on a classical phase space which represents a quantum state. The term
``quasi-probability'' refers to the fact that the function is not a true probability density
as it takes on negative values for some quantum states.  

In the formulation of quantum theory there are many conceptual barriers
to overcome in gaining an intuition for the nature of quantum systems.
The phase space reformulation, however, allows for visualization and other analytical
techniques that are already well understood and applied to classical probability
distributions.  In this picture, much of the conceptual problems are replaced by one: negative
probability.  As such, negativity features prominently in many studies of the differences and transitions between quantum and classical mechanics.  For example, discussions of negativity of the Wigner function have appeared in the context of decoherence \cite{Paz_93}, chaos \cite{habib_98}, nonlocality \cite{kalev_2009}, simulatability \cite{Veitch2012Negative, Veitch2013Efficient,Mari2012Positive} and many more \cite{ferrie_2011}.

Hudson's theorem \cite{Hudson1974When} was the first to characterize the positive Wigner functions\footnote{Hudson's results was generalized to many-particles in reference \cite
  {Soto1983When}. For a recent simple proof, see reference \cite
  {Toft2006Hudsons} and also references therein for other proofs and
  generalizations.
}.  Surprisingly, he found that the Wigner function of a \emph{pure} quantum state is positive if and only if it is a Gaussian distribution in phase space---in quantum optics terminology, a coherent or squeezed state.  More recently, Gross \cite{Gross2006Hudsons} has found an analog of Hudson's theorem for odd dimensional Wigner functions.  He found that the discrete Wigner function of a pure state is positive if and only if it is a stabilizer state, which is a discrete analog of a Gaussian state.  For the positivity of the Wigner function, the question of mixed states was studied in reference \cite{Srinivas1975Some} and later independently in \cite{Brocker1995Mixed}.  Both references independently found that a theorem in classical probability attributed to Bochner \cite{Bochner1933Monotone} and generalization thereof can be used to characterize both the valid Wigner functions and the subset of positive ones.  Surprising at the time was that positive Wigner functions were not limited to the convex hull of Gaussian states.  This fact allows phase space representations to be generalizations of classically efficient simulation schemes \cite{Veitch2012Negative, Veitch2013Efficient,Mari2012Positive}.

The problem is that the Wigner function is but one of an infinitude of phase space representations of quantum theory, many of which have been utilized to understand some aspects of quantum theory from both foundational and practical perspectives \cite{ferrie_2011}.  It behooves us then to generalize the theorems which identified the positive Wigner functions to other phase space representations. One might hope to then generalize the techniques mentioned above involving positive Wigner function to other phase space representations.

But, how much can the quantum Bochner's theorem be generalized? We could
call any theorem that simultaneously determines the positivity of
a quasi-probability distribution and the underlying operator a generalized
quantum Bochner's theorem, but that would be somewhat disingenuous since it misses the beautiful simplicity which Bochner and others identified.  But to demand such symmetry requires a restriction to a set of phase space representation which possesses a kind of generalized Fourier transform.  This will seem natural once the original quantum Bochner's theorem is stated below since it relies on the usual Fourier transform.  To be precise, we show that discrete phase space representations that are built on projective unitary representations of abelian groups, which also have an additional property we call a \emph{projective Fourier frame}, naturally enable a theorem characterizing the valid and positive quasi-probability distributions.

The outline of the remainder of the paper is as follows.  In Section \ref{S:Bochner's theorem} we review quantum Bochner's theorem for both the continuous and discrete Wigner functions.  In Section \ref{S:Frame formalism} we define the frame formalism which generalizes the Wigner function to arbitrary quasi-probability representations of quantum states.  Section \ref{S:main} contains the main results: a generalization of quantum Bochner's theorem, a sufficient condition for discrete quasi-probability representations to admit such a generalization, and a characterization of these quasi-probability representations.  We conclude with Section \ref{S:end}.

\section{Bochner's theorem for Wigner functions\label{S:Bochner's theorem}}

Let us first recall Bochner's theorem in classical probability \cite{Bochner1933Monotone}.    We require a few standard definitions.  Recall that a probability density on $\mathbb R^2$, for example, is a positive function which is normalized to unity:
\begin{align}
\mu(p,q) &\geq 0,\\
\iint_{\mathbb R^2} \mu(p,q) dq dp &= 1.
\end{align}
A related concept is a \emph{positive definite} function, which is one for which the following is satisfied for arbitrary pairs of real numbers $(\zeta_1,\eta_1),\ldots,(\zeta_N,\eta_N)$ and complex numbers $a_1,\ldots,a_N$:
\begin{equation}\label{pos def function}
\sum_{k,k'=1}^N \conj a_{k}a_{k'} \phi(\zeta_{k'}-\zeta_k,\eta_{k'}-\eta_k)\geq0,
\end{equation}
for all positive integers $N$.

A central objection in classical theory and this work is the Fourier transform of $\mu$.  
\begin{definition}\label{def:char}
The Fourier transform of $\mu$,
\begin{equation}
\phi(\zeta,\eta) = \iint_{\mathbb R^2} \mu(p,q) e^{-i(\zeta q +\eta p)} dqdp,
\end{equation}
is called the \emph{characteristic function}.  The same definition will hold for the Fourier transform on groups other than $\mathbb R^2$.
\end{definition}
Given a function $\phi$, what are the necessary and sufficient conditions for its inverse Fourier transform to be a valid probability density?   The answer is given by Bochner's theorem:
\begin{theorem}\label{Bochner's theorem}
The function $\phi$ is a characteristic function---that is, the Fourier transform of some probability density, if and only if the following are satisfied:
\begin{enumerate}
\item $\phi(0,0)=1$,
\item $\phi$ is continuous,
\item $\phi$ is positive definite.
\end{enumerate}
\end{theorem}

In this section we will review the known quantum generalizations of Bochner's theorem.

\subsection{Continuous Wigner Function}

The position operator $Q$ and momentum operator $P$ are the central objects in the formulation of infinite
dimensional quantum theory.  The operators satisfy the canonical
commutation relations $[  Q,   P]=i$.  We naturally would like a joint probability distribution $\mu_\rho(p,q)$ of the values of $Q$ and $P$.  From the postulates of quantum mechanics we have a rule for calculating expectation values.  Noting that characteristic function in Definition \ref{def:char} is the expectation value of $e^{-i(\zeta q +\eta p)}$, we choose a quantization of $p$ and $q$ to obtain
\begin{equation}\label{Wigner characteristic function}
\phi_\rho(\zeta,\eta):=\av{e^{ -i(\zeta Q+\eta P)}}=\Tr(e^{-i (\zeta Q+\eta P)}\rho).
\end{equation}
Since the characteristic function is the Fourier transform of the joint probability distribution, we invert to obtain
\begin{equation}\label{Wigner Fourier invert of characteristic}
 \mu_\rho(p,q)=\frac{1}{(2\pi)^2}\iint_{\mathbb R^2} \Tr(e^{-i(\zeta Q+\eta P)}\rho) e^{ i (\zeta q +\eta p)}d\zeta d\eta,
\end{equation}
which is the now famous Wigner function of $\rho$ \cite{Wigner1932Quantum}.   The Wigner function is both positive and negative for some quantum states.  However, it otherwise behaves as a classical probability density on the classical phase space.  For these reasons, the Wigner function and others like it came to be called \emph{quasi-probability} functions.  The Wigner functions which are positive are easily characterized by Bochner's theorem.  However, not every probability density corresponds to a valid quantum states.  Thus, we need to generalize Bochner's theorem.

The notion of positive definite function can be generalized to the following: a function $\phi$ is called $\gamma$-positive definite if
\begin{equation}\label{pos zeta def function}
\sum_{k,k'}^N \conj a_{k}a_{k'} \phi(\zeta_{k'}-\zeta_k,\eta_{k'}-\eta_k)e^{i\gamma (\zeta_k\eta_{k'}-\zeta_{k'}\eta_k)/2}\geq0.
\end{equation}

If the definition of the characteristic function of a density operator is extended to a generic operator, denoted also by $\rho$, which is not necessarily a density operator to begin with, then a theorem generalizing Bochner's theorem can be stated as follow, wherein the straightforward normalization condition is omitted: \cite{Srinivas1975Some,Brocker1995Mixed}:
\begin{theorem}\label{wigner_bochner}
Let $\phi_\rho$ be the characteristic function of $\rho$.  Then,
\begin{enumerate}
\item $\rho$ is a density operator if and only if $\phi_\rho$ is 1-positive definite.
\item $\rho$ is a density operator with positive Wigner function representation if and only if $\phi_\rho$ is simultaneously 1-positive definite and 0-positive definite.
\end{enumerate}
\end{theorem}

This is the \emph{quantum Bochner's theorem} which we would like to generalize to arbitrary quasi-probability representations.  Note that the Fourier transform features prominently---after all, the theorem is stated not for the quasi-probability density but for its Fourier transform.  One finite dimensional generalization of the Wigner formalism is built on the \emph{discrete} Fourier transform.  As we might expect then, the theorem looks completely analogous for that case.

\subsection{Discrete Wigner Function}

Consider a complex Hilbert space $\mathcal H$ of odd dimension $d$ and denote its standard basis $\{\ket{\psi_k}\}$.  We define the conjugate basis as the discrete Fourier transform of $\{\ket{\psi_k}\}$.  That is,
\begin{equation}
 \ket{\varphi_m}=\frac{1}{\sqrt{d}}\sum_{k=0}^{d-1}e^{ -i 2\pi km/d}\ket{\psi_k}.\label{mombasis}
\end{equation}
Recall that for the continuous case, to arrive at the Wigner function, we Fourier inverted the characteristic function $ \av{e^{ -i(\zeta Q+\eta P)}}$.  First note that due to the Campbell-Baker-Hausdorff formula, this is equivalent to inverting $ \av{e^{-i\eta P}e^{-i\zeta Q}e^{-i \zeta\eta/2}} $.  This formula is not valid when $  Q$ and $  P$ are bounded operators.  We can continue upon noting that  $e^{-i P}$ and $e^{-i Q}$ are the \emph{discrete} displacement operators in the \emph{continuous} functions they act on.  They are often referred to as the \emph{shift} and \emph{boost} operators.

We define the \emph{generalized Pauli matrices} $  Z$ and $  X$ as
\begin{equation}
\begin{aligned}
  Z\ket{\psi_k}=\omega^k\ket{\psi_k},\;\;&  Z\ket{\varphi_m}=\ket{\varphi_{m+1}},\\
  X\ket{\varphi_m}=\omega^{-m}\ket{\varphi_m},\;\;&  X\ket{\psi_k}=\ket{\psi_{k+1}}.
\end{aligned}\label{genPaulis}
\end{equation}
where $\omega=e^{-i2\pi/d}$.  Notice that $  Z$ and $  X$ generate discrete displacements in our discrete functions.  That is, $X$ plays the role of $e^{-i P}$ and $Z$ plays the role of $e^{-i Q}$.  So, instead of Fourier inverting $ \av{e^{-ij P}e^{-il Q}\omega^{jl/2}} $, we discretely Fourier invert what we define as the characteristic function of a density matrix acting on $\mathcal H$:
\begin{equation}
\phi_\rho(j,l):=\av{  X^{j}  Z^{l} \omega^{jl/2}}.\label{disccharZX}
\end{equation}
Since there is no multiplicative inverse of 2 in even dimensions, we can see why this prescription is only valid for odd $d$.  The discrete Fourier inverse is the \emph{discrete Wigner function}
\begin{equation}
\mu_\rho(q,p):=\frac{1}{d}\sum_{j,l=0}^{d-1} \omega^{-(jq+lp)}\Tr(\rho  X^{j}  Z^{l} \omega^{jl/2}).
\end{equation}
If $\rho=\op\psi\psi$ is a pure state and $\ip{\psi_k}{\psi}=\alpha_k$ then 
\begin{equation}
\mu_\rho(q,p)=\frac{1}{d}\sum_{s=0}^{d-1} \omega^{-ps}\alpha_{q-\frac{s}{2}}\alpha_{q+\frac{s}{2}}^\ast\label{GrossDWF}
\end{equation}
is (unitarily equivalent to) the discrete Wigner function defined by Gross \cite{Gross2006Hudsons}.

A Hermitian operator is a quantum state if it is positive semi-definite.  That is, $\rho$ is density matrix if and only if $\av{A^\dag A}=\Tr(\rho A^\dag A)\geq0$ for all Hermitian $A$. Since $\{X^{j}  Z^{l} \omega^{jl/2}\}$ constitutes an orthogonal basis, it can be used to expand $A$ in its \emph{Fourier representation}
\begin{align*}
\av{A^\dag A}&=\sum_{j,l=0}^{d-1}\sum_{j',l'=0}^{d-1} \conj a_{jl} a_{j'l'} \av{Z^{-l}X^{-j}X^{j'}Z^{l'}\omega^{(j'l'-jl)/2}}\\
&=\sum_{j,l=0}^{d-1}\sum_{j',l'=0}^{d-1} \conj a_{jl} a_{j'l'} \phi_\rho(j'-j,l'-l)\omega^{(jl'-j'l)/2}.
\end{align*}
Let $\zeta=(j,l)$.  Then defining the matrices
\begin{align*}
M_{\zeta\zeta'}^C&=\phi_{\rho}(j'-j,l'-l),\\
M_{\zeta\zeta'}^Q&=\phi_{\rho}(j'-j,l'-l)\omega^{(jl'-j'l)/2},\\
\end{align*}
we have the follow discrete analog of quantum Bochner's theorem:
\begin{theorem}\label{gross_bochner}
Let $\phi_\rho$ be the characteristic function of $\rho$.  Then,
\begin{enumerate}
\item $\rho$ is a density operator if and only if $M^Q\geq0$.
\item $\rho$ is a density operator with positive discrete Wigner function representation if and only if both $M^C\geq0$ and $M^Q\geq0$.
\end{enumerate}
\end{theorem}

Notice the similarity to Theorem \ref{wigner_bochner}.  The theorems are nearly identical save for the particular definition of ``characteristic function''.  Next, we outline the generalizations of the Wigner functions to arbitrary quasi-probability representations  and then characterize the additional structure required to make a sensible definition of characteristic functions and hence, Fourier transform.

\section{Arbitrary quasi-probability representations and the frame formalism\label{S:Frame formalism}}
Let $\mathcal H$ be a complex Hilbert space of dimension $d$.  Let $\Lambda$ be some set of cardinality $d^2 \leq \abs{\Lambda} <\infty$.  This set represents some classical ontology (such as a phase space).  In \cite{Ferrie2008Frame}, a \emph{generalized quasi-probability representation} was defined to be any invertible affine map $\mu:\rho\to\mu_\rho$ which has
\begin{equation}
\mu_\rho(\lambda) \in \mathbb R \text{ and } \sum_{\lambda\in\Lambda} \mu_\rho(\lambda) = 1.
\end{equation}
Later, in \cite{Ferrie2009Framed}, this was generalized to include both states and measurements in a classical formalism which only has the constraint of positivity relaxed.  Such a formalism encompasses all known quasi-probability representations \cite{ferrie_2011}.

A \emph{frame} can be thought of as a generalization of an orthonormal basis \cite{Christensen2003Introduction}.  We will define frames not for the Hilbert space $\mathcal H$  but for set of Hermitian operators acting on $\mathcal H$, denoted $\herm$.  With the trace inner product $\ip{  A}{  B}:=\Tr(  A  B)$, $\herm$ forms a \emph{real} Hilbert space itself of dimension $d^2$.  

A frame for $\herm$ is a set of operators $\mathcal F:=\{ F(\lambda)\}\subset\herm$ which satisfies
\begin{equation}\label{def_discrete_frame}
a\norm{A}^2\leq\sum_{\lambda\in\Lambda} \Tr[ F(\lambda){A}]^2\leq b\norm{A}^2,
\end{equation}
for all $A\in\herm$ and some constants $a,b>0$.  This definition generalizes a defining condition for an orthogonal basis $\{ B_k\}_{k=1}^{d^2}$
\begin{equation}\label{def_basis}
\sum_{k=1}^{d^2}\Tr[{B_k}{A}]^2 = \norm{A}^2,
\end{equation}
for all $A\in\herm$.  The mapping $A\mapsto\Tr[ {F(\lambda)}{A}]$ is called a \emph{frame representation} of $\herm$.  It was shown in \cite{Ferrie2008Frame,Ferrie2009Framed} that each quasi-probability representation is associated with a frame.  That is, given any quasi-probability function $\mu$ on $\Lambda$, it must be the case that it can be obtained via the mapping
\begin{equation}\label{frame map}
\rho\mapsto\mu(\lambda)=\Tr(\rho F(\lambda)).
\end{equation}
The converse is also true: given any frame $F$, Equation (\ref{frame map}) defines a quasi-probability function $\mu$ on $\Lambda$.  Thus, we can equivalently discuss properties of the frame elements rather than those of the mapping itself or its image.  

For example, if a quasi-probability representation features only positive functions, then the frame elements must be positive operators.  If some $F(\lambda)$ is not a positive operator, $\mu(\lambda)$ must obtain negative values for some quantum state.  If we consider both states and measurements, then negativity must appear somewhere.  This necessity of negativity is folklore than has recently been formalized in various contexts\footnote{In this context, the \protect \emph {necessity of negativity} was first proven
  by Spekkens \cite {Spekkens2008Negativity}. Direct proofs using the theory of
  frames were given in references \cite {Ferrie2008Frame, Ferrie2009Framed} and
  generalized to infinite dimensional Hilbert spaces in reference \cite
  {Ferrie2010Necessity} (see also Stuple \cite{Stulpe}). The result, in finite dimensions, is also implied by earlier work on a
  related topic in reference \cite {Busch1993Classical}.}.

Since a frame representation is invertible, one can always invert the quasi-probability distribution and study the properties, and/or validity, of the density matrix.  This is somehow unappealing, however, as it would be more convenient to work directly within the phase space formulation one is considering.  This is precisely what Bochner's theorem does.  However, in order to obtain such simplification, some addition structure on the phase space representation must be assumed.  The examples we considered in the previous section already possess quite a bit of symmetry.  Next, we will generalize the necessary structure to a requirement on a generic quasi-probability representation.

\section{Projective Fourier frames and their quantum Bochner's theorem \label{S:main}}

This section presents the main results of the paper.  For our purposes, the key difference between the continuous and discrete Fourier transforms is that the former is done on the group $\mathbb{R}$ whereas the latter is done on a group $\mathbb{Z}_d$. In general, the Fourier transform is available on any locally compact group \cite{Folland1994} and the Bochner's theorem can be generalized to locally compact abelian groups, compact groups, and more \cite{Heyer1977}. Here we will restrict ourselves to finite groups. The phase space, hence the index of the frame in the quasi-probability representation, will be identified with a finite abelian group, $G$. We then define a ``Bochner representation'' to be a quasi-probability representation which admits a quantum Bochner's theorem, which will be essentially the same as the Theorems \ref{wigner_bochner} and \ref{gross_bochner} but stated in terms of the characteristic function defined via the generalized Fourier transform on $G$. This restriction to abelian groups allows us to find a family of frames associated with projective representations that gives rise to Bochner representations.

\subsection{Fourier transform and Bochner's theorem on finite abelian groups\label{S:Fourier transform}}
First, we need a Fourier transform on $G$, which in turn requires the notion of irreducible unitary representations of $G$. A unitary representation $U:g\mapsto U_{g}$ of $G$ is a homomorphism from $G$ to the unitary group $U(d)$ of all $d \times d$ unitary matrices. A representation is irreducible if and only if it has no proper subrepresentation. As an example, every group has a one-dimensional trivial irreducible representation $\mathds{1}:g \mapsto 1, \forall g$. The theory simplifies tremendously when $G$ is a finite abelian group since there are always $|G|$ irreducible representations, all of which are one-dimensional and can be constructed by sending generators to primitive $|G|$th roots of unity. Associated with any unitary representation of
a group $G$ is the unitary character $\chi=\Tr\left(U\right)$.
The definition of the Fourier transform $\tilde{f}$ of a function $f$ on an abelian group $G$ relies on the notion of irreducible characters $\left\{ \chi_{j}\right\} $, $\chi_{j}=\Tr\left(U^{\left(j\right)}\right)$,
associated with all irreducible representations $\left\{ U^{\left(j\right)}\right\} $
of $G$:
\begin{align*}
\tilde{f}_{j} & =\frac{1}{|G|}\sum_{g}\chi_{j}(g)f_{g}.
\end{align*}
For instance, the "zero frequency" component of the Fourier transform comes from choosing $\chi_j$ to be the irreducible character of the trivial representation.

The irreducible characters form a basis of functions on $G$. The character table $\chi_{j}(g)$ is a complex Hadamard matrix: it
is unitary, because of the orthogonality of irreducible characters, and $\left|\chi_{j}(g)\right|=1$
$\forall j,g$. In particular, there is the inverse Fourier
transform
\begin{align*}
f_{g} & =\sum_{j}\overline{\chi_{j}(g)}\tilde{f}_{j}.
\end{align*}
and the complex conjugation simply sends $\chi_j(g)$ to $\chi_j(g^{-1})$. This simple relation is absent in the Fourier transform on a non-abelian group, which requires replacing the unitary characters with all the entries of $U$. In that case the complex conjugation sends an entry of $U_g$ to the same entry not of $U_{g^{-1}}$ but of its transpose. This precludes a non-abelian analogue of our main theorem, Theorem \ref{main_thm}.

The Bochner's theorem adapts without difficulty to this setting:
\begin{theorem} \label{gen_bochner}
A function $\phi$ is a characteristic function (that is, the Fourier transform of a valid probability mass function, the discrete analog of probability density function) if and only if the following are satisfied:
\begin{enumerate}
\item $\phi\left(0\right)=1,$
\item $\phi$ is positive definite.
\end{enumerate}
\end{theorem}
\begin{proof}
Part 1 is the normalization condition which can be checked by direct calculation. For 2, consider 
\begin{align*}
\sum_{j,j'=1}\overline{a}_{j}a_{j'}\tilde{f}_{j'-j} & =\frac{1}{|G|}\sum_{j,j'=1}\overline{a}_{j}a_{j'}\sum_{g}\chi_{j'-j}(g)f_{g}\\
 & =\frac{1}{|G|}\sum_{j,j'=1}\overline{a}_{j}a_{j'}\sum_{g}\chi_{j'}(g)\overline{\chi}_{j}(g)f_{g}\\
 & =\frac{1}{|G|}\sum_{g}f_{g}\left|\sum_{j=1}a_{j}\chi_{j}(g)\right|^{2}.
\end{align*}
If $f_{g}\ge0$ $\forall g$, then the LHS is also positive. Conversely,
if $f_{g}<0$ at $g'$, then we can choose $a_{j}$ so that $\sum_{j=1}a_{j}\chi_{j'}(g)$
vanishes everywhere except at $g=g'$, which is possible because $\left\{ \chi_{j}\right\} $
form a basis of functions on $G$.
\end{proof}

This theorem will give us the analog of the 0-positivity condition for the quasi-probability distributions over $G$.  Next, we will need to define the analog of the 1-positivity condition.  In anticipation of this, we first require some additional structure on the frame elements.  In particular, they must arise from a unitary projective representation.

\subsection{Unitary projective representations\label{S:Projective representation}}

Instead of mapping to the unitary group $U(d)$, a unitary projective representation $\Pi$ is a homomorphism from $G$ to the projective unitary group $PU(d)$. In other words,
$\left\{ \Pi_g\right\} $ is a group with multiplication up to a
\emph{2-cocycle}, $\alpha\left(g,g'\right):G\times G\to\mathbb{C}$,
\begin{align*}
\Pi_{g}\Pi_{g'} & =\alpha\left(g,g'\right)\Pi_{gg'}.
\end{align*}
As an example, suppose that $\{\Pi_g\}$ is the generalized Pauli matrices in $d$ dimensions $\{ X^j Z^l \omega^{jl/2} \}$, one can check by direct computation that it is a unitary projective representation of $\mathbb{Z}_d \times \mathbb{Z}_d$ with $\alpha\left(g,g'\right)=\omega^{(j' l - jl')/2}$.

Two projective representations $\Pi$ and $\Pi'$ are projectively
equivalent if and only if $\Pi_g$ and $\Pi'_g$ are similar matrices
up to an arbitrary complex number $\mu(g)$ with $\mu(e)=1$:
\begin{align*}
\Pi_g & =\mu\left(g\right)U^{-1}\Pi'_gU,\; \forall g.
\end{align*}
This translates into the condition on the 2-cocycles
\begin{align}
\alpha\left(g,g'\right) & =\frac{\mu(g)\mu(g')}{\mu\left(gg'\right)}\alpha'(g,g').\label{eq:cohomologous}
\end{align}
2-cocycles that give rise to equivalent projective representations
are \emph{cohomologous} and belong to the same \emph{2-cohomology
class} $\tilde{\alpha}$. All projective representations of $G$ can
be classified according to their 2-cohomology class $\tilde{\alpha}$.
A canonical reference on this subject is the multi-volume work by Karpilovsky \cite[volumes 2 and 3]{Karpilovsky}.

Within each 2-cohomology class, we will choose to work with a phase,
$\left|\alpha\right|=1$, which can always be done. We will also choose
$\alpha\left(g,g^{-1}\right)=\alpha\left(g^{-1},g\right)=1$ $\forall g$
so that $\Pi_{g}^{-1}=\Pi_{g^{-1}}$ $\forall g$ as is the case, again, for the generalized Pauli matrices. The first equality
always holds because $\Pi_g$ and $\Pi_{g^{-1}}$ commute.
For the second equality, we need to find $\mu$ such that, using (\ref{eq:cohomologous})
with $\alpha'\left(g,g'\right)=1$, we have
\begin{align*}
\alpha\left(g,g^{-1}\right) & =\mu(g)\mu\left(g^{-1}\right).
\end{align*}
It is clear that $\mu(g)$ and $\mu\left(g^{-1}\right)$ can
be chosen independently of the values of $\mu$ at any other group
element. Therefore, with this choice of $\mu$, $\Pi_{g}^{-1}=\Pi_{g^{-1}}$
$\forall g$.

Additionally, to handle cases corresponding to overcomplete frames, we will need to know about the kernel. The kernel of a projective representation is a set of group elements
that are mapped to the identity operator up to a phase. A projective
representation is \emph{faithful} if and only if ker $\Pi=\left\{ e\right\} $,
the identity element. Since ker $\Pi$ is a normal subgroup of $G$,
any $\Pi$ descends to a faithful projective representation $\pi$
of $G/\ker\Pi$ defined by $\pi\left(g\ker\Pi\right)=\Pi_g$
up to a phase and vice versa. So there is a one-one correspondence
between projective representations of $G$ and of $G/\ker\Pi$, which
preserves irreducibility and frameness.

\subsection{Projective Fourier frames and Bochner representations of states}

With the above preparations, we are now ready to discuss conditions on the frame elements which define a quasi-probability representation.
We allow the frame to inherit the Fourier transform from its underlying group:
\begin{definition}Let $\{F_j\}$ be a frame and $G$ be a finite abelian group.  Then $\{\tilde F_g\}_{g\in G}$ such that
\begin{equation}\label{fourier frame}
F_j =\frac1{|G|}\sum_g \chi_j(g) \tilde F_g,
\end{equation}
is called the Fourier frame of $\{F_j\}$ and $\phi_{\rho}(g) = \Tr(\rho \tilde F_g)$ is a characteristic function
of $\rho$.
\end{definition}

The generalization of the classical Bochner's theorem (Theorem \ref{gen_bochner}) immediately applies:
\begin{theorem}\label{quant_class_boch}
A characteristic function $\phi_{\rho}$ of $\rho$
is the Fourier transform of a positive quasi-probability representation
of $\rho$ if and only if the following are satisfied:
\begin{enumerate}
\item $\phi_{\rho}\left(0\right)=1,$
\item $\phi_{\rho}$ is positive definite.
\end{enumerate}
\end{theorem}

Next we define a generalization of the 1-positivity condition such that, together with Theorem \ref{quant_class_boch}, we have a full generalization of quantum Bochner's theorem.
\begin{definition}
Let $G$ be a finite group of size $N$ and $g,g' \in G$. Given $\alpha:G\times G\to\mathbb{C}$,
a function $\phi$ on $G$ is $\alpha$-positive definite
if and only if the matrix defined by
\begin{align*}
M^Q_{gg'} & =\phi \left( g'g^{-1}\right) \alpha \left( g',g^{-1}\right)
\end{align*}
is positive semidefinite.
\end{definition}

Note that we could do away with associativity in the definition (so that $G$ is not a group), but we will not consider $G$ that is not a group in the rest of the paper. So this definition will suffice for our application.

Now recall we would like to find those quasi-probability representations which give rise to a theorem analogous to quantum Bochner's theorem for the Wigner function (Theorem \ref{wigner_bochner}).  We will call such a quasi-probability representation a \emph{Bochner representation} and define it formally as follows:
\begin{definition} Let $\phi_{\rho}$ be a characteristic function of $\rho$.
A \emph{Bochner representation} is a quasi-probability representation in which for some $\alpha$
\begin{enumerate}
\item $\rho$ is a density operator if and only if $\phi_{\rho}$ is $\alpha$-positive definite.
\item $\rho$ is a density operator with positive quasi-probability representation if and only
if $\phi_{\rho}$ is simultaneously $\alpha$-positive definite and positive definite.
\end{enumerate}
\end{definition}

Not all quasi-probability representations will be Bochner representations.  Possessing the niceties of a quantum Bochner's theorem requires a lot of structure.  With hindsight, we have primed ourselves with the mathematical tools necessary to identify this structure.  The following defines those Fourier frames which are also projective representations:

\begin{definition} A {\em projective frame} $\tilde{\Pi}$ is a frame
which is also the image of a projective representation of an abelian
group $G$ with $\tilde{\Pi}_{g}^{-1}=\tilde{\Pi}_{g^{-1}}$.
\end{definition}

Now we state the main theorem which characterizes the quasi-probability representations which possess a quantum Bochner's theorem---that is, a Bochner representation.
\begin{theorem}\label{main_thm}
Any projective frame is a Fourier frame of a Bochner representation, with $\alpha\left(g^{-1},g'\right)$ being the 2-cocycle of the projective frame.\end{theorem}
\begin{proof}
The Fourier transform of $\tilde{\Pi}$ is a set of Hermitian operators
\begin{align*}
F_{j}^{\dagger} & =\frac{1}{d^{2}}\sum_{g}\overline{\chi_{j}(g)}\tilde{\Pi}_{g}^{\dagger}=\frac{1}{d^{2}}\sum_{g}\chi_{j}\left(g^{-1}\right)\tilde{\Pi}_{g^{-1}}=F_{j}.
\end{align*}
Note that this is not the case if $G$ is non-abelian as discussed in section \ref{S:Fourier transform}. Now, since the Fourier transform is invertible, $F$ is always a frame
whenever $\tilde{\Pi}$ is one, so a projective frame is a Fourier
frame of a quasi-probability representation.

That the representation is a Bochner representation follows from the
ability to expand an arbitrary operator $A$ using the projective
frame:
\begin{align*}
\Tr\left(\rho A^{\dagger}A\right) & =\sum_{g,g'}\overline{a}_{g}a_{g'}\Tr\left(\rho\tilde{\Pi}_{g}^{\dagger}\tilde{\Pi}_{g'}\right)\\
 & =\sum_{g,g'}\overline{a}_{g}a_{g'}\Tr\left(\rho\tilde{\Pi}_{g^{-1}}\tilde{\Pi}_{g'}\right)\\
 & =\sum_{g,g'}\overline{a}_{g}a_{g'}\Tr\left(\rho\tilde{\Pi}_{g'g^{-1}}\right)\alpha\left(g^{-1},g'\right)\\
 & =\sum_{g,g'}\overline{a}_{g}a_{g'}\phi_{\rho}\left(g'g^{-1}\right)\alpha\left(g^{-1},g'\right)\ge0\
\end{align*}
if and only if $\rho$ is positive semidefinite.

The generalized Bochner's Theorem \ref{quant_class_boch} then completes the proof that
$F$ defines a Bochner representation.
\end{proof}

Since the generalized Pauli matrices employed to define the discrete Wigner function of Gross form a projective frame. Theorem \ref{main_thm}, therefore, includes theorem \ref{gross_bochner} as a special case.

\subsection{Characterizing and constructing projective Fourier frames}

Now that we have imposed additional structural requirements on the frame, we are left with the question of existence of the projective frame other than the one generated by the generalized Pauli matrices. As we will see, Fourier frames are in fact quite ubiquitous. They can be divided into projective frames that are also faithful projective representations and those that are not. The latter class can provide overcomplete frames for constructing quasi-probability representations, but not always.
\begin{definition}
A faithful projective frame is a frame which is faithful as a projective representation
. Otherwise it is an unfaithful projective frame.
\end{definition}

\begin{theorem}
  
\begin{enumerate}
\item A faithful projective frame $\tilde{\Pi}$ exists if and only if $G$
is a symmetric product of groups $H\times H=:H^{2}$ with $\left|H\right|=d$.
\item Both $\tilde{\Pi}$ and its the Fourier frame are orthogonal frames---that is, orthogonal bases.
\end{enumerate}
\end{theorem}
\begin{proof}
1. Being a frame forces a projective frame to be irreducible as a
projective representation.
Then it is known that $G=H^{2}$ with $\left|H\right|=d$ if and only
if it has a faithful irreducible projective representation \cite[Volume 3, Theorem 8.2.18]{Karpilovsky}. 
It is worth noting that this implies faithful projective frames are exactly the \emph{nice error bases} of quantum error correction with abelian index group \cite{knill96}.

2. We first prove that every $\tilde{\Pi}_{g}$, $g\neq e$, is traceless
by writing it as a commutator. Fix $g\neq e$. $\tilde{\Pi}_{g}$
is not proportional to the identity operator by our assumption of
faithfulness, so $g'$ can be found such that $\left[\tilde{\Pi}_{g},\tilde{\Pi}_{g'}\right]\neq0$
because otherwise $\tilde{\Pi}$ is reducible (for $d>1$). By the
group property, we can write $\tilde{\Pi}_{g}$ as a product of two
non-commuting operators
\begin{align*}
\alpha\left(g',g'^{-1}g\right)\tilde{\Pi}_{g} & =\tilde{\Pi}_{g'}\tilde{\Pi}_{g'^{-1}g}.
\end{align*}
But $G$ is abelian, so
\begin{align*}
0\neq\left[\tilde{\Pi}_{g'},\tilde{\Pi}_{g'^{-1}g}\right] & =\left[\alpha\left(g',g'^{-1}g\right)-\alpha\left(g'^{-1}g,g'\right)\right]\tilde{\Pi}_{g}.
\end{align*}
Take the trace of both sides and the claim is proved.

Consequently, they are all orthogonal in the trace inner product.
Therefore, since $\left|\tilde{\Pi}\right|=d^{2}$, it is an orthogonal
basis of $GL\left(\mathcal{H}\right)$, as is $F$, which can be verified
easily.
\end{proof}

The generating matrices of every representative faithful projective frame up
to $d=7$ are listed at \cite{klappi}.
For general $d$, as long as we only consider representations over the complex field, at least one representative projective representation of each and every 2-cohomology classes of $G$ appears as an ordinary representation of a (non-unique) \emph{covering group} of $G$, which can be found, for instance, by the command \texttt{SchurCover(G)} 
in \texttt{GAP} \cite{GAP}. 

By the one-one correspondence between
projective representations of $G$ and of $G/\ker\tilde{\Pi}$, the task of finding an unfaithful projective frame reduces to
the task of lifting a faithful projective frame of an abelian group $H^{2}$
with $\left|H\right|=d$ to the corresponding projective frame of
a group $G$ which has $\ker\tilde{\Pi}$ as a normal subgroup and $G/\ker\tilde{\Pi}=H^{2}$.
Finding such $G$ is an abelian \emph{extension problem} \cite{extension}. Note that $\ker\tilde{\Pi}$ can be any abelian group.

To summarize, the procedure to construct a Bochner representation is as follow:
\begin{enumerate}
\item Pick an abelian group $H$ of size $d$, the dimension of the quantum system.
\item Extend the group $H^2$ by a (possibly trivial) abelian group to $G$.
\item Construct an irreducible projective representation $\{\Pi\}$ of $G$ up to projective equivalence. Within the equivalence class, each unitary operator is still only defined up to a phase. Choose the phases under the constraint $|\alpha|=1$ and $\Pi_{g}^{-1}=\Pi_{g^{-1}}$ $\forall g$. The set of operators is now a projective frame.
\item Fourier transform the projective frame according to equation (\ref{fourier frame}) to obtain the frame $\{F_j\}$ of a Bochner representation.
\end{enumerate}

\subsection{Examples}

First we illustrate the kind of quasi-probability representations that can arise in the characterization of faithful projective frames by examples in $d=2,3,4$. For $d=2$, there is only $G=\mathbb{Z}_{2}^{2}$ and one 2-cohomology
class with the Pauli matrices $\left\{ I,X,Y,Z\right\} $ as a representative
projective frame. The requirement that $\tilde{F}_{g}^{-1}=\tilde{F}_{g^{-1}}$
constrains $\tilde{F}$ to be Hermitian in addition to being unitary,
leaving us with the choices to put $\pm1$ in front of $X,Y,$ or
$Z$. But since $\chi_{j}(g)=\pm1$ also, upon doing the Fourier transform,
we end up with only two quasi-probability representations that are
not related by a unitary transformation depending on whether we change the phases
 of an odd or even number of Pauli matrices. They coincides with the two
similarity classes of the qubit phase space identified in \cite{Gibbons2004}. (See Figure \ref{wootters}.)

\begin{figure}
\centering
\includegraphics{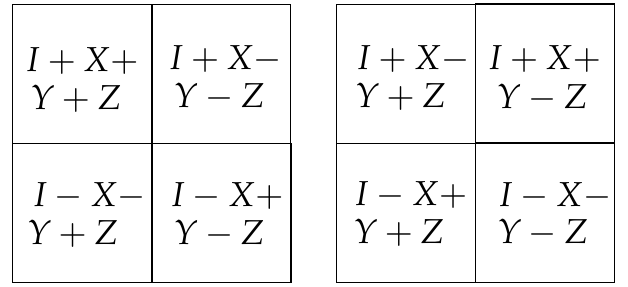}
\caption{The frames for the two similarity classes of qubit phase space in \cite{Gibbons2004}}
\label{wootters}
\end{figure}

For $d=3$, there is only $G=\mathbb{Z}_{3}^{2}$, and two inequivalent
projective representations but with the same image. The phase freedom
that remains after setting $\tilde{F}_{g}^{-1}=\tilde{F}_{g^{-1}}$
is enough to make the set of quasi-probability representations generated
from the two classes identical. The phase freedom, however, supplies a continuum
of choices in choosing different $\tilde{F}$, whose quasi-probability
representations are in general not unitarily related, one of them being the discrete Wigner function of Gross \cite{Gross2006Hudsons}. The case $d=4$ is the first instance with more than one choice of group ($\mathbb{Z}_{4}^{2}$
or $\mathbb{Z}_{2}^{4}$) and inequivalent projective representations
(of $\mathbb{Z}_{2}^{4}$) that generate distinct quasi-probability
representations.

What of unfaithful projective frames? It is well known that transferring Wigner representation
tomography to an even dimensional system requires care.  Wootters first suggested factoring $d$ into its prime components and considering each component as a subsystem.  The discrete phase space is then a multidimensional array, with a dimension for each subsystem \cite{wootters87}.  Alternatively, Leonhardt defined a discrete Wigner function for even $d$ on a single $2d$ by $2d$
phase space \cite{leonhardt}. From the perspective of this work, the reason for doubling the phase space is that the discrete analog of the continuous displacement operators considered there satisfies the relation $\tilde{\Pi}_{g}^{-1}=\tilde{\Pi}_{g^{-1}}$ not in $\mathbb{Z}_d^2$ but in $\mathbb{Z}_{2d}^2$. Leonhardt's phase space is thus an example of a Bochner representation defined by an unfaithful projective frame.

An unfaithful projective frame of size $N$ can give rise to a frame of size less than $N$. As an example, suppose that the quantum system is 2 dimensional and the group is $\mathbb{Z}_2^3$. There is an irreducible projective representation sending the elements (0,0,0),(1,0,0) to $I$, (0,0,1),(1,0,1) to $X$, (0,1,0),(1,1,0) to $Z$, and (0,1,1),(1,1,1) to $Y$. (That is, (1,0,0) is in the kernel of this representation.) But one can find an irreducible character sending the first group element of each pair above to 1 while sending the other element to -1. This component of the Fourier transform is therefore zero, which is not surprising because this projective frame is highly redundant. (All the nonzero components still form a basis.)

\subsection{Non-examples}

There are many discrete analogs of the Wigner function identified in \cite{ferrie_2011} which do not possess the symmetry we have utilized here. For example, Hardy's vector
representation \cite{hardy_2001} and the representation based on symmetric informationally complete positive operator valued measures (SIC POVMs) \cite{zauner,renes} whose frame is a set of rank-one projection operators
that form a basis. There is no projective frame for this representation
since a complete set of projection operators cannot all be pairwise orthogonal.

\section{Conclusion\label{S:end}}

Quantum Bochner's theorem lets us work directly within the phase space formalism of some quasi-probability representation.  This convenience is an alternative to mapping back to the usual complex matrix representation to discuss properties of quantum states.  Working entirely within the phase space representation of Wigner has proven useful for many conceptual and computational tasks in quantum mechanics.
In this paper we have generalized the quantum Bochner's theorem beyond the Wigner function to other discrete phase space distributions.  Specifically, we have found that the quasi-probability representations defined by the Fourier transform of projective representations of finite abelian groups admit a quantum analog of Bochner's theorem.

We hope that this result will prove useful in the quest to identifying classes of quantum states which can be represented positively in some representation \cite{wallman}. With some additional work, we could then, for example, say such states are efficiently simulatable, which can be seen as a generalization of simulation results pertaining only to the Wigner function \cite{Veitch2012Negative, Veitch2013Efficient, Mari2012Positive}.  The goal of this line of reasoning is to have a concrete operational meaning for the oft-quoted sentiment that negativity is an indicator of quantumness.

In future work, we hope to investigate alternative constructions for non-abelian $G$ which preserves the Hermiticity of the frame.  Also, Leonhardt's phase space seems to be the only existing example of an unfaithful projective frame which possesses a quantum Bochner's theorem.  Some further investigation would be required to determine if the endeavor of characterizing all $G$ that give rise to unfaithful projective frames is promising. Finally, investigations have begun in studying negativity in relativistic generalizations of the Wigner function, where it has been found that non-Gaussian states can also possess positive Wigner function \cite{campos}. Perhaps the approach taken here can shed light on the nonexistence of a naive extension of Hudson's theorem to the relativistic setting.

\begin{acknowledgements}
The authors thank Carl Caves and Shashank Pandey for helpful discussions.  This work was supported in part by National Science Foundation Grant No. PHY-1212445.  CF was also supported by the Canadian Government through the NSERC PDF program, the IARPA MQCO program, the ARC via EQuS project number CE11001013, and by the US Army Research Office grant numbers W911NF-14-1-0098 and W911NF-14-1-0103.
\end{acknowledgements}


\begin{thebibliography}{10}

\bibitem{Wigner1932Quantum}
E. Wigner, \emph{On the Quantum Correction For Thermodynamic Equilibrium}, \href{http://dx.doi.org/10.1103PhysRev.40.749}{Physical Review {\bf 40}, 739 (1932)}.

\bibitem{Paz_93}
J.~P. Paz, S. Habib and W.~H. Zurek, \emph{Reduction of the wave packet: Preferred observable and decoherence time scale},\href{http://dx.doi.org/10.1103/PhysRevD.47.488}{Physical Review D {\bf 47}, 488 (1993)}.

\bibitem{habib_98}
S. Habib, K. Shizume and W.~H. Zurek, \emph{Decoherence, Chaos, and the Correspondence Principle}, \href{http://dx.doi.org/10.1103/PhysRevLett.80.4361}{Physical Review Letters {\bf 80}, 4361 (1998)}.

\bibitem{kalev_2009} 
A. Kalev, A. Mann, P.~A. Mello and M. Revzen, \emph{Inadequacy of a classical interpretation of quantum projective measurements via Wigner functions}, \href{http://dx.doi.org/10.1103/PhysRevA.79.014104}{Physical Review A {\bf 79}, 014104 (2009)}.


\bibitem{Veitch2012Negative}
V. Veitch, C. Ferrie, D. Gross and J. Emerson, \emph{Negative quasi-probability as a resource for quantum computation}, \href{http://dx.doi.org/10.1088/1367-2630/14/11/113011}{New Journal of Physics {\bf 14}, 113011 (2012)}.

\bibitem{Veitch2013Efficient}
V. Veitch, N. Wiebe, C. Ferrie and J. Emerson, \emph{Efficient simulation scheme for a class of quantum optics experiments with non-negative Wigner representation},
\href{http://dx.doi.org/10.1088/1367-2630/15/1/013037}{New Journal of Physics {\bf 15}, 013037 (2013)}.

\bibitem{Mari2012Positive}
A. Mari and J. Eisert, \emph{Positive Wigner Functions Render Classical Simulation of Quantum Computation Efficient}, 
\href{http://dx.doi.org/10.1103/PhysRevLett.109.230503}{Physical Review Letters {\bf 109}, 230503 (2013)}.

\bibitem{ferrie_2011}
C. Ferrie, {\em Quasi-probability representations of quantum theory with applications to quantum information science}, \href{http://dx.doi.org/10.1088/0034-4885/74/11/116001}{Reports on Progress in Physics {\bf 74}, 116001 (2011)}.

\bibitem{Hudson1974When}
R.~Hudson,
\newblock {\em When is the Wigner quasi-probability density non-negative?},
\href{http://dx.doi.org/10.1016/0034-4877(74)90007-X}{Reports on Mathematical Physics {\bf 6}, 249 (1974)}.


\bibitem{Soto1983When}
F. Soto and P. Claverie,
\newblock {\em When is the wigner function of multidimensional systems nonnegative?}
\newblock \href{http://dx.doi.org/10.1063/1.525607}{Journal of Mathematical Physics {\bf 24}, 97 (1983)}.

\bibitem{Toft2006Hudsons}
Joachim Toft.
\newblock {\em Hudson's Theorem and Rank One Operators in Weyl Calculus},
\href{http://dx.doi.org/10.1007/3-7643-7514-0_11}{Operator Theory: Advances and Applications {\bf 164}, 153 (2006)}


\bibitem{Gross2006Hudsons}
D.~Gross,
\newblock {\em Hudson's theorem for finite-dimensional quantum systems},
\newblock \href{http://dx.doi.org/10.1063/1.2393152}{Journal of Mathematical Physics {\bf 47}, 122107 (2006)}.

\bibitem{Srinivas1975Some}
M.~D. Srinivas and E.~Wolf,
\newblock {\em Some nonclassical features of phase-space representations of quantum
  mechanics},
\newblock \href{http://dx.doi.org/10.1103/PhysRevD.11.1477}{Physical Review D {\bf 11}, 1477 (1975)}.

\bibitem{Brocker1995Mixed}
T.~Br\"{o}cker and R.~F. Werner,
\newblock {\em Mixed states with positive wigner functions},
\newblock \href{http://dx.doi.org/10.1063/1.531326}{Journal of Mathematical Physics {\bf 36}, 62 (1995)}.

\bibitem{Bochner1933Monotone}
S.~Bochner,
\newblock {\em Monotone funktionen, stieltjessche integrate, und harmonischeanalyse},
\newblock \href{http://dx.doi.org/10.1007/BF01452844}{Mathematische Annalen {\bf 108}, 378 (1933)}.



\bibitem{Christensen2003Introduction}
O. Christensen,
\newblock {\em An Introduction to Frames and Riesz Bases}.
\newblock Birkh\"{a}user,  Boston, 2003.



\bibitem{Ferrie2008Frame}
C. Ferrie and J. Emerson,
\newblock {\em Frame representations of quantum mechanics and the necessity of
  negativity in quasi-probability representations},
\newblock \href{http://dx.doi.org/10.1088/1751-8113/41/35/352001}{Journal of Physics A: Mathematical and Theoretical {\bf
  41}, 352001 (2008)}.

\bibitem{Ferrie2009Framed}
C. Ferrie and J. Emerson,
\newblock {\em Framed hilbert space: hanging the quasi-probability pictures of
  quantum theory},
\newblock \href{http://dx.doi.org/10.1088/1367-2630/11/6/063040}{New Journal of Physics {\bf 11}, 063040 (2009)}.





\bibitem{Spekkens2008Negativity} R.~W. Spekkens, \emph{Negativity and Contextuality are Equivalent Notions of Nonclassicality}, \href{http://dx.doi.org/10.1103/PhysRevLett.101.020401}{Physical Review Letters {\bf 101}, 020401  (2008)}.

\bibitem{Ferrie2010Necessity} C. Ferrie, R. Morris and J. Emerson, \emph{Necessity of negativity in quantum theory}, \href{http://dx.doi.org/10.1103/PhysRevA.82.044103}{Physical Review A {\bf 82}, 044103 (2010)}.

\bibitem{Stulpe} W. Stulpe, \emph{From the attempt of certain classical reformulations of quantum mechanics to quasi-probability representations}, \href{http://dx.doi.org/10.1063/1.4861939}{Journal of Mathematical Physics {\bf 55}, 012109 (2014)}.

\bibitem{Busch1993Classical} P. Busch, K.-E. Hellwig and W. Stulpe, \emph{On Classical Representations of
Finite-Dimensional Quantum Mechanics}, \href{http://dx.doi.org/10.1007/BF00673351}{International Journal of Theoretical Physics {\bf 32}, 399 (1993)}.

\bibitem{Folland1994} G. B. Folland, \emph{A Course in Abstract Harmonic Analysis}. CRC Press, 1994.

\bibitem{Heyer1977} H. Heyer, \emph{Probability Measures on Locally Compact Groups}. Springer-Verlag, 1977.

\bibitem{Karpilovsky} G. Karpilovsky, \emph{Group Representations}, 5 volumes. Elsevier Science, Amsterdam, 1992-1996

\bibitem{knill96}
E. Knill, \emph{Group representations, error bases and quantum codes}, \href{http://arxiv.org/abs/quant-ph/9608049}{LANL report LAUR-96-2807}.


\bibitem{klappi}A. Klappenecker and M. Rotteler, \emph{Beyond stabilizer codes I: Nice error bases},\href{http://dx.doi.org/10.1109/TIT.2002.800471}{IEEE Transactions on Information Theory {\bf 48}, 2392 (2002)} and the \href{http://faculty.cs.tamu.edu/klappi/ueb/ueb.html}{accompanying website}.

\bibitem{GAP} The GAP Group, \emph{Groups,  Algorithms  and Programming}, Version 4.7.5, 
 2014, \url{http://www.gap-system.org}.

V. Dabbaghian, \emph{Repsn---a GAP package}, Version 3.0.2, 2011, \url{http://www.sfu.ca/~vdabbagh/gap/repsn.html}

\bibitem{extension} J. J. Rotman, \emph{An Introduction to the Theory of Groups}, 4th ed. Springer-Verlag, 1995.

\bibitem{Gibbons2004} K. S. Gibbons, M. J. Hoffman, and W. K. Wootters, \emph{Discrete phase space based on finite fields}, \href{http://dx.doi.org/10.1103/PhysRevA.70.062101}{Physical Review A {\bf 70}, 062101 (2004)}.

\bibitem{wootters87}
W.~K. Wootters, \emph{A Wigner-function formulation of finite-state quantum mechanics}, \href{http://dx.doi.org/10.1016/0003-4916(87)90176-X}{Annals of Physics {\bf 176}, 1 (1987)}.

\bibitem{leonhardt}
U. Leonhardt, \emph{Quantum-State Tomography and Discrete Wigner Function}, \href{http://dx.doi.org/10.1103/PhysRevLett.74.4101}{Physical Review Letters {\bf 74}, 4101 (1995)}.

\bibitem{hardy_2001} L. Hardy, \emph{Quantum Theory From Five Reasonable Axioms}, \href{http://arxiv.org/abs/quant-ph/0101012}{arXiv:quant-ph/0101012 (2001)}.

\bibitem{zauner} G. Zauner, \emph{Quantum Designs: Foundations of a Noncommutative Design Theory}, \href{http://dx.doi.org/10.1142/S0219749911006776}{International Journal of Quantum Information {\bf 9}, 445 (2011)}.
\bibitem{renes}
J.~M. Renes, R. Blume-Kohout, A. J. Scott and C. M. Caves, \emph{Symmetric informationally complete quantum measurements}, \href{http://dx.doi.org/10.1063/1.1737053}{Journal of Mathematical Physics {\bf 45}, 2171 (2004)}.

\bibitem{wallman}
J.~J. Wallman and S.~D. Bartlett, {\em Non-negative subtheories and quasiprobability representations of qubits}, \href{http://dx.doi.org/10.1103/PhysRevA.85.062121}{Physical Review A {\bf 85}, 062121 (2012)}.

\bibitem{campos}
A.~G. Campos, R. Cabrera, D.~I. Bondar and H.~A. Rabitz, {\em Violation of Hudson's theorem in relativistic quantum mechanics}, \href{http://dx.doi.org/10.1103/PhysRevA.90.034102}{Physical Review A {\bf 90}, 034102 (2014)}.

\end{thebibliography}
\end{document}